\documentclass[conference]{IEEEtran}
\IEEEoverridecommandlockouts

\usepackage{url}
\usepackage{cite}
\usepackage{amsmath,amssymb,amsfonts}
\usepackage{amsthm}
\usepackage{algorithm}
\usepackage{algpseudocode}
\usepackage{caption}
\makeatletter
\newenvironment{breakablealgorithm}
  {\par\addvspace{\topsep}
   \refstepcounter{algorithm}
   \noindent
   \hrule height .8pt depth 0pt
   \kern 2pt
   \renewcommand{\caption}[2][\relax]{%
     {\raggedright\textbf{\ALG@name~\thealgorithm} ##2\par}%
     \ifx\relax##1\relax
       \addcontentsline{loa}{algorithm}
         {\protect\numberline{\thealgorithm}##2}%
     \else
       \addcontentsline{loa}{algorithm}
         {\protect\numberline{\thealgorithm}##1}%
     \fi
     \kern 2pt\hrule\kern 2pt
   }}
  {\kern 2pt\hrule\relax
   \par\addvspace{\topsep}}
\makeatother

\usepackage{graphicx}
\usepackage{textcomp}
\usepackage{xcolor}
\newtheorem{Definition}{Definition}
\newtheorem{Lemma}{Lemma}
\newtheorem{Theorem}{Theorem}
\usepackage{comment}
\newtheorem{Remark}{Remark}
\def\BibTeX{{\rm B\kern-.05em{\sc i\kern-.025em b}\kern-.08em
    T\kern-.1667em\lower.7ex\hbox{E}\kern-.125emX}}

\begin{document}

\title{Insdel Reed-Solomon Decoding Algorithms up to the Half-Singleton Bound}

\author{\IEEEauthorblockN{Bencheng Fan}
\IEEEauthorblockA{\textit{Dept. EE Engineering} \\
\textit{The University of Melbourne}\\
Melbourne, Australia \\
benchengf@student.unimelb.edu.au}
\and
\IEEEauthorblockN{Margreta Kuijper}
\IEEEauthorblockA{\textit{Dept. EE Engineering} \\
\textit{The University of Melbourne}\\
Melbourne, Australia \\
mkuijper@unimelb.edu.au}
\and
\IEEEauthorblockN{Jingge Zhu}
\IEEEauthorblockA{\textit{Dept. EE Engineering} \\
\textit{The University of Melbourne}\\
Melbourne, Australia \\
jingge.zhu@unimelb.edu.au}
}

\maketitle

\begin{abstract}
 In this paper, we focus on decoding optimal Reed-Solomon codes of dimension two for insdel errors up to the half-Singleton bound. We present a decoding algorithm for the general insertion-and-deletion scenario with a time complexity smaller than brute force decoding. Our algorithm corrects insdel errors up to the half-Singleton bound, which exceeds the scope of existing methods.  
 Additionally, for the special cases of insertion-only and deletion-only insdel error patterns, we present two different algorithms with a time complexity smaller than our general algorithm. 
\end{abstract}

\section{Introduction}

DNA storage is an emerging data storage technology that uses synthetic DNA molecules to encode and store digital information. In a DNA-based storage system, information is encoded into sequences over the DNA alphabet $\{A, G, C, T \}$. Because of its high data density, energy efficiency, and longevity, DNA-based storage is a promising approach for meeting the ever-growing demands of data storage \cite{ieee_dna_storage_2024}\cite{9504568}.  \\

Despite these advantages, DNA-based storage faces several technological challenges, including high costs, slow read and write speeds, and high error rates in the synthesis and sequencing processes. These errors stem from mutations, degradation, and sequencing errors. More specifically, DNA molecules are susceptible to insertions and deletions (insdel errors)\cite{heckel2019characterization}\cite{8937735}\cite{shomorony2022information}.\\

The study of codes that are resistant to insdel errors was pioneered by Levenshtein, Varshamov and Tenengolts in the 1960s \cite{levenshtein1967asymptotically}\cite{levenshtein1966binary}. In \cite{levenshtein1966binary}, Levenshtein showed that any code capable of correcting $t$ deletions can correct $t$ insertions and deletions.\\

 It is demonstrated in \cite{cheng2023efficient} that the number of insdel errors that an $[n,k]$ linear code can correct is at most $n-2k+1$, which is called the half-Singleton bound. Linear codes that reach the half-Singleton bound are called optimal. Recent literature has shown that it is possible to construct optimal Reed-Solomon codes. 
Chen and Zhang \cite{chen2022improved} provide an explicit construction of $[n,2]$ Reed-Solomon codes that can correct up to the half-Singleton bound. Con \textit{et al.}\cite{con2023reed} generalize this result by providing an explicit construction of an optimal $[n,k]$ Reed-Solomon code for $k\geq 2$ and $n\geq 2k$. Then Con \textit{et al.}\cite{con2024optimal} provide a different construction for optimal $[n,2]$ Reed-Solomon codes. 

So far, the problem of constructing optimal Reed-Solomon codes has been addressed. Moreover, designing efficient decoding algorithms for Reed-Solomon codes against insdel errors has also gained attention in recent years. The recent paper by Singhvi \cite{singhvi2026optimally} addresses this problem by providing an efficient decoding algorithm with time complexity of $\mathcal{O}(n^3)$ that decodes optimal $[n,2]$ Reed-Solomon codes for up to $n-3$ deletions. It also provides a decoding algorithm that can recover any $n-3$ deletions in $\mathcal{O}(n)$ time for a constrained subclass of optimal $[n,2]$ Reed-Solomon codes presented in \cite{con2024optimal}. However, \cite{singhvi2026optimally} is limited to the deletions-only scenario. For $[n,k]$ optimal RS codes in general where $k\geq 2$, Banerjee \textit{et al.}\cite{banerjee2025decoding} propose a polynomial time decoding algorithm that can correct up to $t$ insdel errors as long as $t\cdot k =\mathcal{O}(n)$. We observe that a necessary condition for the method proposed in \cite{banerjee2025decoding} to correct $t$ insdel errors is that $t<(k+1)n-\sqrt{(k^2+2k)n^2+kn}$, which is significantly less than the half-Singleton bound $n-2k+1$. The detailed deduction is in the appendix.

\subsection{Our Results}

In this paper, for optimal $[n,2]$ Reed-Solomon codes, our results are 
three-fold. Firstly, we present a decoding algorithm that corrects up to $n-3$ deletions with a time complexity of $\mathcal{O}(n^3)$, which has the same time complexity as the deletion-only algorithm recently proposed by Singhvi\cite{singhvi2026optimally}. Secondly, we present a decoding algorithm that corrects up to $n-3$ insertions with a time complexity of $\mathcal{O}(n^3)$.  Finally, we present a decoding algorithm against general scenarios of insdel errors with a time complexity of $\mathcal{O}(n^5)$.  Whereas the approach of Singhvi \cite{singhvi2026optimally} focuses on deletion-only error patterns, our algorithm handles arbitrary combinations of insertions and deletions. Moreover, our method corrects up to the half-Singleton bound, which exceeds the scope of \cite{banerjee2025decoding}. 
Although we focus on optimal $[n,2]$ Reed-Solomon codes in this paper, the decoding algorithms presented can be easily generalised to non-optimal $[n,2]$ MDS codes.

\subsection{Organization}
The rest of the paper is organised as follows. In Section II, we introduce some key definitions and known results. In Section III, we discuss the decoding algorithms for optimal $[n,2]$ Reed-Solomon codes up to the half-Singleton bound. In Section IV, we briefly summarise the results presented in this paper and discuss potential directions for further study.

\section{Preliminaries}
For a sequence $\mathbf{x}$ with a length of at least $m$ and a vector of indices $\mathbf{J}=(j_1,j_2,\cdots,j_m)$, we denote the subsequence $(x_{j_1},x_{j_2},\cdots,x_{j_m})$ as $\mathbf{x_J}$. For simplicity, we denote the subsequence with only two entries $(x_{j_1},x_{j_2})$ as $\mathbf{x}(j_1,j_2)$. For a matrix $G\in\mathbb{F}_q^{k\times n}$, we denote its column vectors as $\mathbf{g}_i$, where $i\in\{1,2,\cdots,n\}$. For a vector of indices $\mathbf{J}$, we denote the $k\times m$ matrix $[\mathbf{g}_{j_1},\mathbf{g}_{j_2},\cdots,\mathbf{g}_{j_m}]$ as $G_J$.

\begin{Definition}
    Let $\mathbf{x}$ be a sequence, if a sequence $\mathbf{y}$ can be obtained by deleting certain symbols from $\mathbf{x}$, then $\mathbf{y}$ is called a subsequence of $\mathbf{x}$. If a sequence $\mathbf{y}$ is a subsequence of both $\mathbf{x_1}$ and $\mathbf{x_2}$, then $\mathbf{y}$ is called a common subsequence of $\mathbf{x_1}$ and $\mathbf{x_2}$.
\end{Definition}

\begin{Definition}
    For two sequences $\mathbf{x},\mathbf{y}$, the insdel distance $d(\mathbf{x},\mathbf{y})$ is defined as the minimum number of insertions and deletions to transform $\mathbf{x}$ into $\mathbf{y}$.
\end{Definition}

Similar to the minimum Hamming distance, we define the minimum insdel distance of a code.
\begin{Definition}
    For a code $C$, its minimum insdel distance $d(C)$ is defined as

\begin{equation}
  d(C)=\min_{\mathbf{c_1},\mathbf{c_2}\in C,\mathbf{c_1}\neq \mathbf{c_2}}\{d(\mathbf{c_1},\mathbf{c_2})\}  
\end{equation}
\end{Definition}

In \cite{levenshtein1966binary}, Levenshtein demonstrated that the insdel distance is a metric. He also showed that the minimum insdel distance is the key indicator of a code's capability to correct insdel errors.
\begin{Lemma}
\label{Insdel}
(\cite{levenshtein1966binary}) Let $\mathcal{C}$ be a code over $\mathbb{F}_q$. The code $\mathcal{C}$ can correct $t$ insdel errors if and only if $2t<d(C)$.
\end{Lemma}

Clearly, calculating the insdel distance is essential for assessing a code's ability to correct insdel errors. It is well known (see e.g.\cite[Lemma 1]{singhvi2026optimally}) that the insdel distance between two sequences can be calculated via their longest common
subsequence. We present their relationship in the following lemma.
\begin{Lemma}
\label{LCS}
     Let $\mathbf{a}\in\mathbb{F}_q^{n_1}$ and $\mathbf{b}\in\mathbb{F}_q^{n_2}$, then we have
\begin{equation}
    d(\mathbf{a},\mathbf{b})=n_1+n_2-2LCS(\mathbf{a},\mathbf{b}),
\end{equation}
where $LCS(\mathbf{a},\mathbf{b})$ denotes the length of the longest common subsequence of $\mathbf{a}$ and $\mathbf{b}$. 
\end{Lemma}
Lemma \ref{LCS} translates
the problem of computing the insdel distance between two sequences to determining the length of their longest common subsequence. In \cite{hirschberg1977algorithms}, an algorithm that solved the longest common subsequence problem in $\mathcal{O}(n^2)$ time was presented. Below, we give the formal definition of a Reed-Solomon code.
\begin{Definition}
    Let $\alpha_1$, $\alpha_2$, $\cdots$, $\alpha_n\in \mathbb{F}_q$ be distinct evaluation points in a finite field $\mathbb{F}_q$ of order $q\geq n$. For $k\leq n$ the $[n,k]$ RS-code defined by the evaluation set $\{\alpha_1,\alpha_2,\cdots,\alpha_n\}$ is a set of codewords $\{\mathbf{c_f}=(f(\alpha_1),f(\alpha_2),\cdots,f(\alpha_n))|f\in\mathbb{F}_q[x],deg f<k\}$. The generator matrix of this code is shown below:
\begin{equation}
  G=\begin{bmatrix}
1 & 1 & \cdots & 1\\
\alpha_1 & \alpha_2 & \cdots & \alpha_n\\
\alpha_1^2 & \alpha_2^2 & \cdots & \alpha_n^2\\
\vdots & \vdots & & \vdots\\
\alpha_1^{k-1} & \alpha_2^{k-1} & \cdots & \alpha_n^{k-1}\\
\end{bmatrix}  
\end{equation}
\end{Definition}

As a benchmark for comparison, we first present the brute force method for decoding optimal $[n,k]$ Reed-Solomon codes against insdel errors and then give an analysis of its time complexity.
\begin{algorithm}[H]
\caption{Brute-Force Decoding Method}
\label{Brute}

\begin{algorithmic}[1]
\Require Received sequence $\mathbf{y}$ of length $\ell$,
         generator matrix $G$, and code length $n$
\Ensure Decoded codeword $\mathbf{c}_o$ or \textsc{false}

\For{each $\mathbf{u} \in \mathbb{F}_q^k$}
    \State $\mathbf{c} \gets \mathbf{u}\cdot G$

    \If{$d(\mathbf{c},\mathbf{y}) \leq n-2k+1$}
        \State $\mathbf{c}_o \gets \mathbf{c}$
        \State \Return $\mathbf{c}_o$
    \EndIf
\EndFor

\State \Return \textsc{false}
\end{algorithmic}
\end{algorithm}

\begin{Theorem}
\label{Theo: Brute}
    Algorithm \ref{Brute} correctly decodes any optimal $[n,k]$ Reed-Solomon code for up to $n-2k+1$ insdel errors in $\mathcal{O}(q^k\cdot n^2)$ time. Moreover, if $k=2$, Algorithm \ref{Brute} has a time complexity of $\mathcal{O}(n^8)$.
\end{Theorem}
\begin{proof}
    An optimal $[n,k]$ Reed-Solomon code can correct up to $n-2k+1$ insdel errors. In other words, there exists exactly one codeword $\mathbf{c_o}\in \mathcal{C}$ such that $d(\mathbf{c_o},\mathbf{y})\leq n-2k+1$, provided that the number of insdel errors does not exceed the half-Singleton bound. By calculating the insdel distance between each of the codewords and the received sequence $\mathbf{y}$, the decoder will find the unique codeword $\mathbf{c_o}$ such that $d(\mathbf{c_o},\mathbf{y})\leq n-2k+1$ and return it as the output. The method's time complexity is calculated as follows. The decoder may have to go through the whole codebook to find $\mathbf{c_o}$. As the message words are chosen from the set $\mathbb{F}_q^k$, the size of the codebook is $|\mathcal{C}|=q^k$. For each word
    $\mathbf{u}$, the decoder needs to calculate the corresponding codeword by $\mathbf{c}=\mathbf{u}\cdot G$, which takes $\mathcal{O}(n)$ time. In addition, it takes $\mathcal{O}(n^2)$ time to calculate the insdel distance $d(\mathbf{c},\mathbf{y})$ \cite{hirschberg1977algorithms}. Therefore, the time complexity of the brute-force method can be calculated as
    $q^k\cdot \big(\mathcal{O}(n)+ \mathcal{O}(n^2)\big)=\mathcal{O}(q^k\cdot n^2)$. For the moreover part, the smallest known field size for optimal $[n,2]$ Reed-Solomon codes is $q=(n+1)^3$ \cite{con2024optimal}. Correspondingly, the time complexity of Algorithm \ref{Brute} is $\mathcal{O}((n^3)^2)\cdot \big(\mathcal{O}(n) + \mathcal{O}(n^2)\big)=\mathcal{O}(n^8)$.
\end{proof}

\section{Main Results}
The number of insertions and deletions are denoted as $r$ and $s$, respectively. 
In this section, we present decoding algorithms for optimal $[n,2]$ Reed-Solomon codes for decoding insdel errors up to the half-Singleton bound. In other words, the number of insdel errors $r+s$ must be less than or equal to $n-3$. Note that, upon receiving the sequence $\mathbf{y}$, we do not know the values of $r$ and $s$, but we do know the length $\ell$ of $\mathbf{y}$. 

\subsection{Decoding $[n,2]$ RS-Codes Against Deletions}
In this subsection, we present an algorithm for optimal $[n,2]$ Reed-Solomon codes for the deletion-only scenario. Although there is existing literature solving this exact problem with the same time complexity\cite{singhvi2026optimally}, we present our alternative algorithm to complement the results of the rest of the paper.
\begin{breakablealgorithm}
\caption{Algorithm for the Deletion-Only Scenario}
\label{Deletion}

\begin{algorithmic}[1]
\Require Received sequence $\mathbf{y}$ of length $\ell$,
         generator matrix $G$, and code length $n$
\Ensure Decoded codeword $\mathbf{c}_o$ or \textsc{false}

\If{$\ell < 3$}
    \State \Return \textsc{false}
\EndIf

\State $\mathbf{I} \gets (1,2)$

\For{all $(j_1,j_2)$ such that $1 \leq j_1 < j_2 < n$}
    \State $\mathbf{J} \gets (j_1,j_2)$
    \State Solve $\mathbf{u}\cdot G_{\mathbf{J}}=\mathbf{y}_{\mathbf{I}}$
    \State $\mathbf{c} \gets \mathbf{u}\cdot G$

    \If{$\mathbf{y}$ is a subsequence of $\mathbf{c}$}
        \State $\mathbf{c}_o \gets \mathbf{c}$
        \State \Return $\mathbf{c}_o$
    \EndIf
\EndFor

\State \Return \textsc{false}
\end{algorithmic}
\end{breakablealgorithm}

\begin{Theorem}
\label{Theo: Deletion}
    Algorithm \ref{Deletion} correctly decodes any optimal $[n,2]$ Reed-Solomon codes for up to $n-3$ deletions in  $\mathcal{O}(n^3)$ time. If the number of insertions is not zero, or the number of deletions exceeds $n-3$, then Algorithm \ref{Deletion} returns FALSE.
\end{Theorem}

\begin{proof}
    Let $\mathbf{c_o}$ be the true codeword. In a
    deletion-only scenario, $\mathbf{y}(1,2)$ must be a subsequence of $\mathbf{c_o}$, say $\mathbf{c_o}(j_1,j_2)$. As we assume that the number of deletions is at most $n-3$, we have that $j_2\leq2+(n-3)=n-1$. Therefore, $j_1$ and $j_2$ satisfy $1\leq j_1<j_2 <n$. Let $\mathbf{I}=(1,2)$ and $\mathbf{J}=(j_1,j_2)$. Denote the solution of the linear equation $\mathbf{u}\cdot\mathbf{G_J}=\mathbf{y_I}$ by $\mathbf{u}$.  Suppose that $\mathbf{y}$ is a subsequence of $\mathbf{c}=\mathbf{u}\cdot G$. We now show by contradiction that $\mathbf{c}$ must be equal to $\mathbf{c_o}$. The Reed-Solomon code has a minimum insdel distance of $d(C)=2n-4$. Suppose that $\mathbf{c}\neq \mathbf{c_o}$, then by Lemma \ref{LCS} $\mathbf{c}$ and $\mathbf{c_o}$ do not share a common subsequence of length $3$. Since $\mathbf{y}$ has a length $\geq 3$ and is a subsequence of both $\mathbf{c}$ and $\mathbf{c_o}$, it follows that $\mathbf{c}=\mathbf{c_o}$.
    For each pair $(j_1,j_2)$, the encoding process in step 8 takes $\mathcal{O}(n)$ time. It is easy to see that checking whether $\mathbf{y}$ is a subsequence of $\mathbf{c}$ also takes $\mathcal{O}(n)$ time. Since there are $\frac{(n-1)(n-2)}{2}$ combinations of $(j_1,j_2)$, the total time complexity is $\mathcal{O}(n^2)\cdot\big(\mathcal{O}(n)+\mathcal{O}(n)\big)=\mathcal{O}(n^3)$. If all possible combinations of $(j_1,j_2)$ fail, then either the number of insertions is not zero or the number of deletions exceeds $n-3$. The algorithm returns FALSE.
\end{proof}

\subsection{Decoding $[n,2]$ RS-Codes Against Insertions}
Unlike the deletion-only scenario, for efficiently decoding optimal $[n,2]$ Reed-Solomon codes, no existing literature addresses the insertion-only scenario. Therefore, in this subsection, we present a novel algorithm addressing this problem.

\begin{algorithm}[H]
\caption{Algorithm for the Insertion-Only Scenario}
\label{Insertion}

\begin{algorithmic}[1]
\Require Received sequence $\mathbf{y}$ of length $\ell$,
         generator matrix $G$, and code length $n$
\Ensure Decoded codeword $\mathbf{c}_o$ or \textsc{false}

\If{$\ell > 2n-3$}
    \State \Return \textsc{false}
\EndIf

\State $\mathbf{J} \gets (1,2)$

\For{all $(i_1,i_2)$ such that $1 \leq i_1 < i_2 < n$}
    \State $\mathbf{I} \gets (i_1,i_2)$
    \State Solve $\mathbf{u}\cdot G_{\mathbf{J}}=\mathbf{y}_{\mathbf{I}}$
    \State $\mathbf{c} \gets \mathbf{u}\cdot G$

    \If{$\mathbf{c}$ is a subsequence of $\mathbf{y}$}
        \State $\mathbf{c}_o \gets \mathbf{c}$
        \State \Return $\mathbf{c}_o$
    \EndIf
\EndFor

\State \Return \textsc{false}
\end{algorithmic}
\end{algorithm}

\begin{Theorem}
\label{Theo: Insertion}
    Algorithm \ref{Insertion} correctly decodes any optimal $[n,2]$ Reed-Solomon codes for up to $n-3$ insertions in  $\mathcal{O}(n^3)$ time. If the number of deletions is not zero, or the number of insertions exceeds $n-3$, then Algorithm \ref{Insertion} returns FALSE.
\end{Theorem}

\begin{proof}
    Let $\mathbf{c_o}$ be the true codeword. In an insertion-only scenario, $\mathbf{c_o}(1,2)$ must be a subsequence of $\mathbf{y}$, say $\mathbf{y}(i_1,i_2)$.  As we assume that the number of insertions is at most $n-3$, we have that $i_2\leq2+(n-3)=n-1$. Therefore, $i_1$ and $i_2$ satisfy $1\leq i_1<i_2 <n$. Let $\mathbf{J}=(1,2)$ and $\mathbf{I}=(i_1,i_2)$. Denote the solution of the linear equation $\mathbf{u}\cdot\mathbf{G_J}=\mathbf{y_I}$ by $\mathbf{u}$.  Suppose $\mathbf{c}=\mathbf{u}\cdot G$ is a subsequence of $\mathbf{y}$. We now show by contradiction that $\mathbf{c}$ must be equal to $\mathbf{c_o}$. Suppose that $\mathbf{c}\neq \mathbf{c_o}$. Since both $\mathbf{c}$ and $\mathbf{c_o}$ are subsequences of $\mathbf{y}$, we have $d(\mathbf{c},\mathbf{y})=d(\mathbf{c_o},\mathbf{y})=\ell-n$. Therefore, by the triangle inequality, $d(\mathbf{c},\mathbf{c_o})\leq d(\mathbf{c},\mathbf{y})+d(\mathbf{c_o},\mathbf{y})\leq 2\ell-2n\leq2n-6$, contradicting $d(C)=2n-4$. It follows that $\mathbf{c}= \mathbf{c_o}$. For each pair $(i_1,i_2)$, the encoding process in step 8 takes $\mathcal{O}(n)$ time. Checking whether $\mathbf{c}$ is a subsequence of $\mathbf{y}$ also takes $\mathcal{O}(n)$ time. Since there are $\frac{(n-1)(n-2)}{2}$ combinations of $(i_1,i_2)$, the total time complexity is $\mathcal{O}(n^2)\cdot\big(\mathcal{O}(n)+\mathcal{O}(n)\big)=\mathcal{O}(n^3)$. If all possible combinations of $(i_1,i_2)$ fail, then either the number of deletions is not zero or the number of insertions exceeds $n-3$. The algorithm returns FALSE.
\end{proof}

\subsection{Decoding $[n,2]$ RS-Codes Against Insdel Errors}
We now present the main result of the paper, which is an algorithm for decoding optimal $[n,2]$ Reed-Solomon codes for up to $n-3$ insdel errors in general.

\begin{algorithm}[H]
\caption{Algorithm for the General Insdel Scenario}
\label{General}
\begin{algorithmic}[1]
\Require Received sequence $\mathbf{y}$ of length $\ell$,
         generator matrix $G$, and code length $n$
\Ensure Decoded codeword $\mathbf{c}_o$ or \textsc{false}

\If{$\ell<3$ or $\ell>2n-3$}
    \State \Return \textsc{false}
\EndIf

\For{all $(j_1,j_2)$ such that $1\le j_1<j_2\le n$}
    \State $\mathbf{J}\gets(j_1,j_2)$
    \For{$i=1$ to $\ell-1$}
        \State $\mathbf{I}\gets(i,i+1)$
        \State Solve $\mathbf{u}\cdot G_{\mathbf{J}}=\mathbf{y}_{\mathbf{I}}$
        \State $\mathbf{c}\gets\mathbf{u}\cdot G$
        \If{$d(\mathbf{c},\mathbf{y})\le n-3$}
            \State $\mathbf{c_o}\gets\mathbf{c}$
            \State \Return $\mathbf{c_o}$
        \EndIf
    \EndFor
\EndFor

\State \Return \textsc{false}
\end{algorithmic}
\end{algorithm}

To facilitate the proof of Algorithm \ref{General}'s correctness, we introduce the following Lemma.

\begin{Lemma}
\label{adjacent}
    Let $\mathbf{c}\in \mathbb{F}_q^n$ be a sequence sent by the transmitter. Let $\mathbf{y}$ be the received sequence of length $\ell$ obtained by deleting $s$ symbols from and then inserting $r$ symbols into $\mathbf{c}$, where $r+s\leq n-3$. Then, there exists $1\leq i\leq \ell-1$ and $1\leq j_1<j_2\leq n$ such that $\mathbf{c}(j_1,j_2)=\mathbf{y}(i,i+1)$. 
\end{Lemma}
\begin{proof}
    After the $s$ deletions, we obtain an intermediate sequence $\mathbf{x}$ with a length of $n-s$. The sequence $\mathbf{x}$ consists of symbols from $\mathbf{c}$. To separate each symbol in $\mathbf{x}$ from each other, at least $n-s-1$ inserted symbols are needed. However, since $r+s\leq n-3$, we have that $r\leq n-s-3<n-s-1$. Therefore, after the $r$ insertions, there must be at least two symbols from $\mathbf{c}$ that are adjacent to each other, even if they are not necessarily adjacent in $\mathbf{c}$. Therefore, we conclude that there exists $1\leq i\leq \ell-1$ and $1\leq j_1<j_2\leq n$ such that $\mathbf{c}(j_1,j_2)=\mathbf{y}(i,i+1)$. 
\end{proof}

We now provide a proof for Algorithm \ref{General} and give an analysis of its time complexity.  
\begin{Theorem}
\label{Theo: General}
     Algorithm \ref{General} correctly decodes any optimal $[n,2]$ Reed-Solomon codes for up to $n-3$ insdel errors in  $\mathcal{O}(n^5)$ time. 
\end{Theorem}

\begin{proof}
    Let $\mathbf{c_o}$ be the true codeword. By Lemma \ref{adjacent}, there exist $1\leq i \leq \ell-1$ and $1\leq j_1<j_2\leq n$ such that $\mathbf{y}(i,i+1)=\mathbf{c_o}(j_1,j_2)$. For each possible combination of $(i,j_1,j_2)$, let $\mathbf{I}=(i,i+1)$ and $\mathbf{J}=(j_1,j_2)$. Denote the solution of the linear equation $\mathbf{u}\cdot\mathbf{G_J}=\mathbf{y_I}$ as $\mathbf{u}$. Suppose that $\mathbf{c}=\mathbf{u}\cdot G$ satisfies $d(\mathbf{c},\mathbf{y})\leq n-3$. We now show by contradiction that $\mathbf{c}$ must be equal to $\mathbf{c_o}$. Suppose that $\mathbf{c}\neq \mathbf{c_o}$. Clearly, $d(\mathbf{c_o},\mathbf{y})\leq n-3$. Since the insdel distance is a metric, $d(\mathbf{c,\mathbf{c_o}})\leq d(\mathbf{c},\mathbf{y})+d(\mathbf{c_o},\mathbf{y})\leq 2n-6$, which contradicts the fact that $d(C)=2n-4$. It follows that $\mathbf{c}= \mathbf{c_o}$. If $d(\mathbf{c},\mathbf{y})>n-3$, the algorithm moves on to the next combination of $(i,j_1,j_2)$. If no codeword satisfying $d(\mathbf{c},\mathbf{y})\leq n-3$ is found for all combinations of $(i,j_1,j_2)$, then it indicates that the number of insdel errors exceeds the half-Singleton bound. The algorithm returns FALSE. The time complexity of Algorithm \ref{General} is calculated as follows. For each combination of $(i,j_1,j_2)$, it takes $\mathcal{O}(n)$ time for the encoding process in step 9. It takes further $\mathcal{O}(n^2)$ time to calculate $d(\mathbf{c},\mathbf{y})$\cite{hirschberg1977algorithms}. Since there are $\frac{n(n-1)(\ell-1)}{2}$ possible combinations of $(i,j_1,j_2)$, the total time complexity is $\big(\mathcal{O}(n)+\mathcal{O}(n^2)\big)\cdot \mathcal{O}(n^3)=\mathcal{O}(n^5)$. 
\end{proof}

\section{Conclusions}
In this paper, we focus on efficient decoding algorithms for optimal $[n,2]$ Reed-Solomon codes up to the half-Singleton bound. The literature so far has produced a decoding algorithm for the deletion-only scenario with a time complexity of $\mathcal{O}(n^3)$\cite{singhvi2026optimally}. Our main result is an algorithm for decoding insdel errors in general with a time complexity of $\mathcal{O}(n^5)$. It improves on the brute-force method with a time complexity of $\mathcal{O}(n^8)$. Additionally, we presented two algorithms with a time complexity of $\mathcal{O}(n^3)$ for deletion-only and insertion-only scenarios, respectively.  We consider our work as a first step toward developing efficient insdel decoding algorithms for $[n,k]$ optimal Reed-Solomon codes.

\bibliographystyle{plain}
\bibliography{references}

\newpage
\appendix
\begin{Theorem}
\label{Theo: Appendix}
Let $\mathcal{C}$ be an $[n,k]_q$ RS code. The deterministic insdel unique-decoding algorithm described in Corollary 1 of \cite{banerjee2025decoding} can correct up to $t$ insdel errors only if
\begin{equation}
    t\leq (k+1)n-\sqrt{(k^2+2k)n^2+kn}.
\end{equation}
\end{Theorem}
\begin{proof}
The condition for the existence of the decoding algorithm is expressed by the following inequality.
\begin{equation}
    t\leq n-\sqrt{(1+\epsilon)\cdot kn\cdot(2t+1)}
\end{equation}

Moving the terms in (9) gives:
$$\sqrt{(1+\epsilon)\cdot kn\cdot(2t+1)}\leq n-t $$
Since it is stated that $\epsilon>0$, $(1+\epsilon)$ must be greater than $1$. Therefore, we have that 
$$\sqrt{kn\cdot(2t+1)} <\sqrt{(1+\epsilon)\cdot kn\cdot(2t+1)}\leq n-t$$
i.e. 
\begin{equation}
    \sqrt{kn\cdot(2t+1)} < n-t.
\end{equation}
Squaring both sides of (10) gives:
$$kn\cdot(2t+1)<(n-t)^2=t^2-2nt+n^2$$
Moving the terms and completing the square gives:
$$t^2-2(k+1)nt+(k+1)^2n^2>(k^2+2k)n^2+kn$$
i.e.
$$[t-(k+1)n]^2>(k^2+2k)n^2+kn$$
Since $t$ must be less than $n$, we have that:
\begin{equation}
    t<(k+1)n-\sqrt{(k^2+2k)n^2+kn}.
\end{equation}
\end{proof}
\begin{Remark}
Given a positive integer $k\geq 2$, as $n$ goes to infinity, we have that:
$$\frac{t}{n}<(k+1)-\sqrt{(k^2+2k)}=\frac{1}{(k+1)+\sqrt{(k^2+2k)}}<\frac{1}{2k}$$
By comparison, an optimal $[n,k]$ RS code can correct up to $t=n-2k+1$ insdel errors. Asymptotically, we have:
$$\lim_{n\rightarrow\infty}\frac{n-2k+1}{n}=1.$$
In conclusion, the method proposed in \cite{banerjee2025decoding} applies only to scenarios where the number of insdel errors is significantly below the upper limit allowed by the half-singleton bound.
\end{Remark}

\end{document}